\documentclass[10pt, conference]{IEEEtran}

\usepackage{amsmath,amssymb,amsthm,graphicx,enumerate,verbatim,xcolor}
\usepackage[top=1in,bottom=1in,right=0.75in,left=0.75in]{geometry}
\usepackage{tikz}
\usetikzlibrary{arrows,positioning,fit,backgrounds}
\usetikzlibrary{calc}

\newcommand{\eps}{\varepsilon}
\newcommand{\ovr}{\overline}
\newcommand{\Ebb}{\mathbb{E}}
\newcommand{\Pbb}{\mathbb{P}}

\newcommand{\Acal}{\mathcal{A}}

\newcommand{\Mcal}{\mathcal{M}}

\newcommand{\Rcal}{\mathcal{R}}
\newcommand{\Scal}{\mathcal{S}}
\newcommand{\Tcal}{\mathcal{T}}

\newcommand{\Xcal}{\mathcal{X}}
\newcommand{\Ycal}{\mathcal{Y}}
\newcommand{\Zcal}{\mathcal{Z}}

\newtheorem{thm}{Theorem}

\newtheorem{lemma}{Lemma}
\newtheorem{defn}{Definition}

\title{Source-Channel Secrecy with Causal Disclosure}
\author{
\authorblockN{Curt Schieler, Eva C. Song, Paul Cuff, H. Vincent Poor}
\authorblockA{Dept. of Electrical Engineering,\\
Princeton University,
Princeton, NJ 08544.\\
E-mail: \{schieler, csong, cuff, poor\}@princeton.edu }
}
\begin{document}
\maketitle 
\begin{abstract}
Imperfect secrecy in communication systems is investigated. Instead of using equivocation as a measure of secrecy, the distortion that an eavesdropper incurs in producing an estimate of the source sequence is examined. The communication system consists of a source and a broadcast (wiretap) channel, and lossless reproduction of the source sequence at the legitimate receiver is required. A key aspect of this model is that the eavesdropper's actions are allowed to depend on the past behavior of the system. Achievability results are obtained by studying the performance of source and channel coding operations separately, and then linking them together digitally. Although the problem addressed here has been solved when the secrecy resource is shared secret key, it is found that substituting secret key for a wiretap channel brings new insights and challenges: the notion of weak secrecy provides just as much distortion at the eavesdropper as strong secrecy, and revealing public messages freely is detrimental.
\end{abstract}
\section{Introduction}
There are a variety of ways to model the presence of secrecy in a communication system. Shannon considered the availability of secret key shared between the transmitter (Alice) and receiver (Bob), using it to apply a one-time pad to the message \cite{Shannon1949}. Wyner introduced the idea of physical-layer security with the wiretap channel and secrecy capacity, exploiting the difference in the channels to Bob and Eve (the eavesdropper) \cite{Wyner1975}. Maurer derived secrecy by assuming that Alice, Bob, and Eve have access to correlated random variables \cite{Maurer1993}. In such models, the measure of security is usually the conditional entropy, or ``equivocation'', of the message; maximum equivocation corresponds to perfect secrecy. 

In this work, we replace equivocation with an operationally motivated measure of secrecy. We want to design our coding and encryption schemes so that if Eve tries to reproduce the source sequence, she will suffer a certain level of distortion. More precisely, the measure of secrecy is the minimum average distortion attained by the cleverest (worst-case) eavesdropper. Occasionally, we will refer to Eve's minimum average distortion as the payoff: Alice and Bob want to maximize the payoff over all code designs. 

\begin{figure}
\begin{tikzpicture}
[node distance=1cm,minimum height=7mm,minimum width=14mm,arw/.style={->,>=stealth'}]
  \node[coordinate] (source) {};
  \node[rectangle,draw,rounded corners] (alice) [right =9mm of source] {Alice};
  \node[rectangle,draw,rounded corners] (ch) [right =9mm of alice] {$P_{Y,Z|X}$};
  \node[rectangle,draw,rounded corners] (bob) [right =of ch,yshift=8mm] {Bob};
  \node[rectangle,draw,rounded corners] (eve) [right =of ch,yshift=-8mm] {Eve};
  \node[coordinate] (shat) [right =of bob] {};
  \node[coordinate] (t) [right =of eve] {};
  \node[rectangle] at ([xshift=4mm,yshift=8.5mm] bob.center) {$i=1,\ldots,n$};

  \draw [arw] (source) to node[midway,above,yshift=-1mm]{$S^n$} (alice);
  \draw [arw] (alice) to node[midway,above,yshift=-1mm]{$X^n$} (ch);
  \draw [arw] (ch.15) to node[pos=0.3,above,yshift=-1mm]{$Y^n$} (bob.west);
  \draw [arw] (ch.345) to node[pos=0.3,below]{$Z^n$} (eve.west);
  \draw [arw] (bob) to node[midway,above,yshift=-.5mm]{$\widehat{S}_i$} (shat);
  \draw [arw] (eve) to node[midway,above,yshift=-1mm]{$T_i$} (t);
  \draw [arw] (bob) to node [midway,right,xshift=-2.5mm,yshift=0.5mm] {$\widehat{S}^{i-1}$} (eve); 
  \draw [dashed] ([xshift=-1cm,yshift=-5mm] eve.center) rectangle ([xshift=1.8cm,yshift=6mm] bob.center);
\end{tikzpicture}
 \caption{\emph{Joint source-channel operations.} Bob and Eve produce the sequences $\widehat{S}^n$ and $T^n$, respectively, after they receive the broadcast channel outputs. In the $i$th step, Bob produces $\widehat{S}_i$, and Eve produces $T_i$ based on her observation of $Z^n$ and $\widehat{S}^{i-1}$.}
\label{source_channel_fig}
\end{figure}
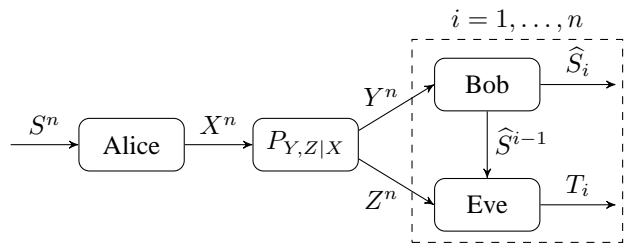

This measure of secrecy has been considered previously by Yamamoto in \cite{Yamamoto1997}, but our setup differs from \cite{Yamamoto1997} in a few ways, the most salient of which is ``causal disclosure''. We assume Bob, the legitimate receiver, is producing actions that are revealed publicly (in particular, to Eve) in a causal manner. We might view Eve as an adversary who is trying to predict Bob's current and future actions based on both the actions that she has already witnessed and the output of her wiretap, and subsequently act upon her predictions (see Figure \ref{source_channel_fig}).

To further motivate the causal disclosure feature of our model, consider the effect of removing it; that is, consider Yamamoto's problem in \cite{Yamamoto1997}. Alice must communicate a source sequence losslessly to Bob over a noiseless channel and the secrecy resource is shared secret key. Eve observes the message but is \emph{not} given causal access to Bob's reproduction sequence. As shown in \cite{Schieler}, the solution to this problem is that any positive rate of secret key is enough to cause Eve unconditionally maximal distortion -- we see that secrecy is alarmingly inexpensive.  Similarly, if the secrecy resource is physical-layer security instead of shared secret key, it can be shown that a positive secrecy capacity is enough to force maximal distortion. However, with causal disclosure in play a tradeoff between secrecy capacity and payoff emerges. As a side remark, one might observe that causal disclosure is consistent with the spirit of Kerckhoffs's principle \cite{Kerchoff}.

In addition to assuming causal disclosure, we further diverge from \cite{Yamamoto1997} by considering the presence of a wiretap channel instead of shared secret key; the problem of shared secret key with causal disclosure was solved in \cite{Cuff2010}. In \cite{Cuff2010}, it was found that the optimal tradeoff between the rate of secret key and payoff earned is achieved by constructing a message that effectively consists of a fully public part and a fully secure part. More specifically, the optimal encoder publicly reveals a distorted version of the source sequence and uses the secret key to apply a one-time pad to the supplement. With this insight, we use the broadcast channel to effectively create two separate channels -- one public and one secure~-- 
\newgeometry{top=0.75in,bottom=1in,right=0.75in,left=0.75in} \noindent 
as is done in \cite{Csiszar--Korner1978}. However, we modify \cite{Csiszar--Korner1978} by removing the requirement that Eve must decode the public message, thereby rendering the message public only in intention, not in reality. We show that freely giving the ``public'' message away dramatically decreases Eve's distortion; thus, her equivocation of the public message becomes important. Upon using channel coding to transform the broadcast channel into effective public and secure channels, we show that the weak secrecy provided by the channel encoder allows us to use the source code from \cite{Cuff2010} to link the source and channel coding operations together digitally. It turns out that a strong secrecy guarantee would not improve the results. Separation allows us to obtain a lower bound on the achievable payoff. We also provide an example of the lower bound and obtain an upper bound. Most of the proofs are omitted.

Before proceeding, we briefly juxtapose our measure of secrecy with equivocation. Equivocation does not give much insight into the structure of Eve's knowledge and does not depend on the actions that Eve makes. In contrast, looking at Eve's distortion tells us something about the quality of her reconstruction if her aim is to replicate a source sequence and produce actions accordingly. There are other instances in the literature where an operational definition of security is used. For example, in \cite{Merhav--Arikan1999}, Merhav et al. looked at the expected number of guesses needed to correctly identify a message.
\section{Problem Statement}
The system under consideration, shown in Figure \ref{source_channel_fig}, operates on blocks of length $n$. 
All alphabets are finite, and both the source and channel are memoryless. The communication flow begins with Alice, who observes a sequence $S^n$, i.i.d. according to $P_S$, and produces an input $X^n$ to the broadcast channel. The memoryless broadcast channel is characterized by $P_{Y,Z|X}$, but since the relevant calculations only involve the marginals $P_{Y|X}$ and $P_{Z|X}$, we need only consider $P_{Y,Z|X}=P_{Y|X}P_{Z|X}$. At one end of the channel, Bob receives $Y^n$ and generates a sequence of actions, $\hat{S}^n$, with the requirement that the block error probability be small. At the other terminal of the channel, Eve receives $Z^n$ and generates $T^n$; in generating $T_i$, she has access to the full $Z^n$ sequence and the past actions of Bob, $\hat{S}^{i-1}$. In essence, we can view Bob and Eve as playing a public game that commences after they receive the channel outputs. In each move, they are allowed to see each other's past moves and produce an estimate of the next source symbol accordingly. Since Bob's reproduction must be almost lossless, his moves are restricted and he does not benefit from knowing Eve's past actions. For similar reasons, revealing $\hat{S}^{i-1}$ to Eve at step $i$ has exactly the same consequences as revealing $S^{i-1}$; henceforth, we consider the causal disclosure to be $S^{i-1}$. A more general version of the game would allow for distortion in Bob's estimate (see \cite{Cuff2010Allerton}), but in this work we focus on lossless communication.

In the next two definitions, refer to Figure \ref{source_channel_fig} for an illustration of the setup.
\begin{defn}
 For blocklength $n$, a source-channel code consists of an encoder $f$ and a decoder $g$:
 \begin{IEEEeqnarray*}{l}
  f:\Scal^n\rightarrow \Xcal^n \text{ (more generally, }P_{X^n|S^n}\text{)}\\
  g: \Ycal^n\rightarrow \Scal^n.
 \end{IEEEeqnarray*}
\end{defn}
Note that we do not restrict the encoder to be deterministic. 

For any source-channel code, we can calculate the probability of block error and the payoff earned against the worst-case adversary, as defined in the following: 
\begin{defn}
 Fix a value function (or, distortion measure) $\pi:\Scal\times\Tcal\rightarrow [0,\infty)$. We say that a payoff $\Pi$ is achievable if there exists a sequence of source-channel codes such that
\begin{equation}
 \label{errdefn}\lim_{n\rightarrow\infty}\Pbb[S^n\neq\hat{S}^n]=0
\end{equation}
and
\begin{equation}
 \label{paydefn}\lim_{n\rightarrow\infty}\min_{\{t_i(s^{i-1},z^n)\}_i}\Ebb\left[\frac1n \sum_{i=1}^n\pi(S_i,t_i(S^{i-1},Z^n))\right]\geq\Pi.
\end{equation}
\end{defn}
Eve's average distortion, i.e. the LHS of (\ref{paydefn}), is defined exactly as in rate-distortion theory for separable distortion measures. Although it is not explicit in (\ref{paydefn}), we assume that Eve has full knowledge of the source-channel code and the source distribution $P_S$.
\section{Lower Bound}
The first result is a lower bound on the maximum achievable payoff.

\begin{thm}
\label{innerbnd}
Fix $P_S$, $P_{Y,Z|X}$, and $\pi(s,t)$. A payoff $\Pi$ is achievable if the inequalities
\begin{IEEEeqnarray*}{rCl}
 I(S;U) &<& I(V;Y) \\
 H(S|U) &<& I(W;Y|V)-I(W;Z|V)\\
 \Pi &\leq& \min_{t(u)}\Ebb[\pi(S,t(U))]
\end{IEEEeqnarray*}
hold for some distribution $P_SP_{U|S}P_VP_{W|V}P_{X|W}P_{Y,Z|X}$.
\end{thm}
Theorem~\ref{innerbnd} is obtained in part by transforming (via channel coding) the noisy broadcast channel into noiseless public and secure channels. However, the result can be strengthened considerably by taking into account Eve's equivocation of the public message. The source-channel code used to achieve Theorem~\ref{innerbnd} remains the same; only the analysis is strengthened.  We illustrate and discuss this further in section~\ref{example}, where we give a brief proof sketch. Our main result is the following theorem.
\begin{thm}
 \label{improved}
 Fix $P_S$, $P_{Y,Z|X}$, and $\pi(s,t)$. A payoff $\Pi$ is achievable if the inequalities
\begin{IEEEeqnarray*}{rCl}
 I(S;U) &<& I(V;Y) \\
 H(S|U) &<& I(W;Y|V)-I(W;Z|V)\\
 \Pi &\leq& \alpha \cdot \Pi_{\max} + (1-\alpha) \cdot \min_{t(u)}\Ebb[\pi(S,t(U))] 
\end{IEEEeqnarray*}
hold for some distribution $P_SP_{U|S}P_VP_{W|V}P_{X|W}P_{Y,Z|X}$, where
$$\Pi_{\max} = \min_t\Ebb[\pi(S,t)]$$ and 
$$\alpha=\frac{[I(V;Y)-I(V;Z)]^+}{I(S;U)}.$$
\end{thm}

We obtain the lower bound in Theorem~\ref{improved} by concatenating a source code and a channel code and matching the rates of the two codes. We first describe what constitutes a good channel code, and the secrecy guarantees that come with it.

\subsection{Channel Code}
The channel code is made up of an encoder and decoder as shown in Figure \ref{channel_fig}. 
\begin{figure}
\begin{tikzpicture}
[node distance=1cm,minimum height=7mm,minimum width=14mm,arw/.style={->,>=stealth'}]
  \node[rectangle,draw,rounded corners] (fc) {$f_c$};
  \node[coordinate] (mp) [left =of fc,yshift=2mm] {};
  \node[coordinate] (ms) [left =of fc,yshift=-2mm] {};
  \node[rectangle,draw,rounded corners] (ch) [right =of fc] {$P_{Y,Z|X}$};
  \node[rectangle,draw,rounded corners] (bob) [right =of ch,yshift=6mm] {$g_c$};
  \node[coordinate] (eve) [right =of ch,yshift=-6mm] {};
  \node[coordinate] (mphat) [right =of bob,yshift=2mm] {};
  \node[coordinate] (mshat) [right =of bob,yshift=-2mm] {};

  \draw [arw] (mp) to node[midway,above,yshift=-1mm]{$M_p$} (mp -| fc.west);
  \draw [arw] (ms) to node[midway,below,yshift=1mm]{$M_s$} (ms -| fc.west);
  \draw [arw] (fc) to node[midway,above,yshift=-1mm]{$X^n$} (ch);
  \draw [arw] (ch.15) to node[midway,above,yshift=-2mm]{$Y^n$} (bob.west);
  \draw [arw] (ch.345) to node[midway,above,yshift=-.8mm]{$Z^n$} (eve.west);
  \draw [arw] (bob.east |- mphat) to node[midway,above,yshift=-1mm]{$\widehat{M}_p$} (mphat);
  \draw [arw] (bob.east |- mshat) to node[midway,below,yshift=1mm]{$\widehat{M}_s$} (mshat);
\end{tikzpicture}
\caption{Channel coding operations.}
\label{channel_fig}
\end{figure}
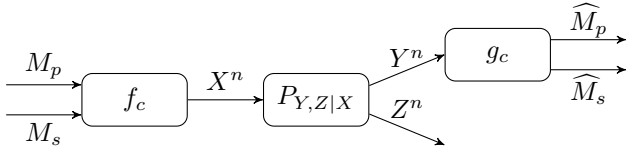
The input to the encoder is a pair of messages $(M_p,M_s)$ destined for the channel decoder, with $M_p$ representing a public message and $M_s$ a secure message. The channel decoder outputs the pair $(\widehat{M}_p,\widehat{M}_s)$. We allow the channel encoder to use private randomization.
\begin{defn}
 A $(R_p,R_s,n)$ channel code consists of a channel encoder $f_c$ and channel decoder $g_c$:
 \begin{IEEEeqnarray*}{l}
  f_c:\Mcal_p\times\Mcal_s\rightarrow \Xcal^n \text{ (more generally, }P_{X^n|M_p,M_s}\text{)}\\
  g_c:\Ycal^n \rightarrow \Mcal_p\times\Mcal_s,
 \end{IEEEeqnarray*}
 where $|\Mcal_p|=2^{nR_p}$ and $|\Mcal_s|=2^{nR_s}$.
\end{defn}
Keeping in mind criteria (\ref{errdefn}) and (\ref{paydefn}) and our desire to form public and private channels, we might ask: what constitutes a good channel code? First, the legitimate channel decoder must recover $M_p$ and $M_s$ with vanishing probability of error. Second, we need a guarantee that we have indeed created a private channel. Ideally, we want to guarantee that the \emph{a priori} distribution on $M_s$ matches the \emph{a posteriori} distribution given both $Z^n$ and $M_p$. If this holds, we are assured that even if the adversary is able to view the public channel perfectly (i.e., recover $M_p$), his optimal strategy for determining $M_s$ is to choose a random message according to the a priori distribution $P_{M_s}$. Later we will exploit the adversary's inability to exactly recover $M_p$, but for now we suppose that it is freely available. We turn to the notion of secrecy capacity and cast our requirement in terms of entropy: we want $H(M_s) \approx H(M_s|Z^n,M_p)$, or $I(M_s;Z^n,M_p)\approx0$. More precisely, the \emph{normalized} mutual information $\frac{1}{n}I(M_s;Z^n,M_p)$ should vanish in a good channel code. Although this measure of secrecy is so-called ``weak secrecy'', it turns out that having strong secrecy would not improve the payoff for the source encoder that we use.

We make a further technical requirement (c.f. \cite{Csiszar--Korner1978}) that good channel codes must satisfy for our purposes, considering the particular source encoder that we employ. The channel code must work not only for $(M_p,M_s)$ independent and uniformly distributed, but more generally in the case that, conditioned on $M_p$, $M_s$ is almost uniform. To be precise, we require
\begin{equation}
 \label{condunif}\max_{m_p,m_s,m_s'}\frac{\Pbb[M_s=m_s|M_p=m_p]}{\Pbb[M_s=m_s'|M_p=m_p]}\leq 2^{n\cdot\delta_n}
\end{equation}
to hold for some $\delta_n$ such that $\delta_n\rightarrow 0$ as $n\rightarrow\infty$. The source encoder we employ will produce message pairs $(M_p,M_s)$ that satisfy this condition, regardless of the source distribution.

\begin{defn}
 The pair of rates $(R_p,R_s)$ is achievable if, for all $(M_p,M_s)$ satisfying (\ref{condunif}) for every $n$, there exists a sequence of $(R_p,R_s,n)$ channel codes such that
\begin{equation*}
 \lim_{n\rightarrow\infty}\Pbb[(M_p,M_s)\neq(\widehat{M}_p,\widehat{M}_s)]=0
\end{equation*}
and
\begin{equation*}
  \lim_{n\rightarrow\infty}\frac1n I(M_s;Z^n,M_p)=0.
\end{equation*}
\end{defn}
The following theorem gives an achievable region, the proof of which comes from modifying the work done in \cite{Csiszar--Korner1978}. The idea is to include enough private randomness in the channel encoder so that the adversary effectively uses his full decoding capabilities to resolve the randomness, leaving no room to additionally decode part of the secret message. The amount of randomness required is the mutual information provided by the adversary's channel.
\begin{thm}
\label{bcc}
The pair $(R_p,R_s)$ is achievable if
\begin{IEEEeqnarray*}{rCl}
 R_p &<& I(V;Y),\\
 R_s &<& I(W;Y|V)-I(W;Z|V)
 \end{IEEEeqnarray*}
for some $P_VP_{W|V}P_{X|W}P_{YZ|X}$.
\end{thm}

\subsection{Source Code}
A source code consists of a source encoder and decoder. The encoder observes a memoryless source $P_S$ and produces a pair of messages, $(M_p,M_s)$, with $M_p$ representing a public message and $M_s$ a secure message. We allow the source encoder to use private randomization.
\begin{defn}
A $(R_p,R_s,n)$ source code consists of an encoder $f_s$ and a decoder $g_s$: 
\begin{IEEEeqnarray*}{l}
 f_s:\Scal^n\rightarrow \Mcal_p\times\Mcal_s \text{ (more generally, }P_{M_p,M_s|S^n}\text{)}\\
 g_s:\Mcal_p\times\Mcal_s \rightarrow \Scal^n,
\end{IEEEeqnarray*}
where $|\Mcal_p|=2^{nR_p}$ and $|\Mcal_s|=2^{nR_s}$.
\end{defn}
As shown in Figure \ref{source_fig}, the output of the source encoder is effectively passed through a channel $P_{Z^n|M_p,M_s}$. 
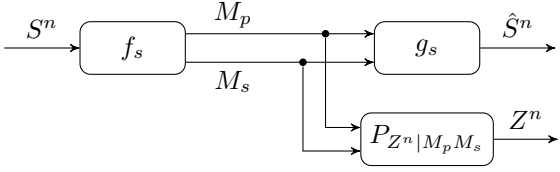
\begin{figure}
\begin{tikzpicture}
[node distance=1cm,minimum height=7mm,minimum width=14mm,arw/.style={->,>=stealth'}]
  \node[coordinate] (source) {};
  \node[rectangle,draw,rounded corners] (fs) [right =of source] {$f_s$};
  \node[circle,fill=black,minimum size = 1mm,inner sep=0pt] (dummy1) [right =1.5cm of fs.345] {};
  \node[circle,fill=black,minimum size = 1mm,inner sep=0pt] (dummy2) [right =1.8cm of fs.15] {};
  \node[rectangle,draw,rounded corners] (gs) [right =2.5cm of fs] {$g_s$};
  \node[coordinate] (shat) [right =of gs] {};
  \node[rectangle,draw,rounded corners] (ch) [below =0.5cm of gs] {$P_{Z^n|M_pM_s}$};
  \node[coordinate] (zn) [right =8.5mm of ch] {};

  \draw [arw] (source) to node[midway,above,yshift=-1mm]{$S^n$} (fs);
  \draw [arw] (fs.15) to node[near start,above,yshift=-1mm]{$M_p$} (fs.15 -| gs.165);
  \draw [arw] (fs.345) to node[near start,below,yshift=1mm]{$M_s$} (fs.345 -| gs.195);
  \draw [arw] (dummy1) |- (ch.190); 
  \draw [arw] (dummy2) |- (ch.170); 
  \draw [arw] (gs) to node[midway,above,yshift=-.5mm]{$\hat{S}^n$} (shat);
  \draw [arw] (ch) to node[midway,above,yshift=-1mm]{$Z^n$} (zn);
\end{tikzpicture}
\caption{Source coding operations.} 
\label{source_fig}
\end{figure}
In light of the previous subsection, we want to consider sequences of channels $\{P_{Z^n|M_p,M_s}\}_{n=1}^\infty$ that provide weak secrecy when the output of the source encoder satifies (\ref{condunif}).
\begin{defn}
 Define $C_S$ to be the set of $\{P_{Z^n|M_p,M_s}\}_{n=1}^\infty$ such that, for all $P_{M_p,M_s}$ satisfying (\ref{condunif}) for every $n$, 
 \begin{equation*}
  \frac1n I(M_s;Z^n|M_p)\rightarrow0. 
 \end{equation*}
\end{defn}
We can view $C_S$ as the resource of physical-layer security. Notice that a sequence of good channel codes yields a sequence of channels in $C_S$ to the adversary. We now consider what payoff can be achieved if rates are imposed on the messages, and $C_S$ is imposed. By considering the availability of a structure $C_S$ and a noiseless channel from Alice to Bob, we are effectively divorcing the goals of source and channel coding so that each can be considered separately.
\begin{defn}
 Fix $P_S$ and $\pi(s,t)$. The triple $(R_p,R_s,\Pi)$ is achievable if there exists a sequence of $(R_p,R_s,n)$ source codes such that
\begin{equation}
 \lim_{n\rightarrow\infty}\Pbb[S^n\neq\hat{S}^n]=0,
\end{equation}
\vspace{0pt}
\begin{equation}
 (M_p,M_s) \text{ satisfies (\ref{condunif}) }\forall n,\vspace{5pt}
\end{equation}
and, for all $\{P_{Z^n|M_p,M_s}\}_{n=1}^\infty\in C_S$,
\vspace{5pt}
\begin{equation}
 \lim_{n\rightarrow\infty}\min_{\{t_i(s^{i-1},z^n)\}_i}\Ebb\left[\frac1n \sum_{i=1}^n\pi(S_i,t_i(S^{i-1},Z^n))\right]\geq\Pi.
\end{equation}
\vspace{3pt}
\end{defn}
We give a region of achievable $(R_p,R_s,\Pi)$. In \cite{Cuff2010}, the region is characterized when the secrecy resource is shared secret key instead of physical-layer security.
\begin{thm}
 \label{sourcethm}
 Fix $P_S$ and $\pi(s,t)$. Then $(R_p,R_s,\Pi)$ is achievable if the inequalities
\begin{IEEEeqnarray*}{rCl}
 R_p &>& I(S;U)\\
 R_s &>& H(S|U)\\
 \Pi &\leq& \min_{t(u)}\Ebb[\pi(S,t(U))]
\end{IEEEeqnarray*}
hold for some $P_SP_{U|S}$.
\end{thm}
The lower bound in Theorem~\ref{innerbnd} follows from Theorems~\ref{bcc} and \ref{sourcethm}. The main idea in the proof of Theorem~\ref{sourcethm} is to use the public message to specify a sequence $U^n$ that is correlated with $S^n$, and use the secure message to encode the supplement that is needed to fully specify the source sequence. The source encoder is defined in such a way that, conditioned on the public message $M_p$, the adversary views the source as if it were generated by passing $U^n$ through a memoryless channel $P_{S|U}$. With this perspective, the past $S^{i-1}$ will no longer help the adversary; Eve's best strategy is to choose a function $t$ that maps $U_i$ to $T_i$. 

Although we omit the full proof of Theorem~\ref{sourcethm}, we provide the crucial lemma that shows how the weak secrecy provided by a good channel code is used in analyzing the payoff. The result of Lemma~\ref{connect} (below) is that we can view Eve as having full knowledge of $M_p$ and $S^{i-1}$ and no knowledge of $M_s$, which fulfills our goal of creating a secure channel and a public channel. In other words, using channel coding to create physical-layer security in the form of a structure $C_S$ allows us to show that, from Eve's perspective, knowledge of $(Z^n,S^{i-1})$ is no more helpful than $(M_p,S^{i-1})$ in easing the distortion. To parse the statement of the lemma, simply look at the arguments of $t(\cdot)$. 
\begin{lemma}
\label{connect}
If $P_{M_p,M_s}$ satisfies (\ref{condunif}) for every $n$, and \\$\{P_{Z^n|M_p,M_s}\}_{n=1}^\infty\in C_S$, then for all $\eps>0$,
\begin{IEEEeqnarray*}{rCl}
 \IEEEeqnarraymulticol{3}{l}{
\min_{t(i,s^{i-1},z^n)}\Ebb\left[\frac1n \sum_{i=1}^n\pi(S_i,t(i,S^{i-1},Z^n))\right]
}\\ \qquad
&\geq&\min_{t(i,s^{i-1},m_p)}\Ebb\left[\frac1n \sum_{i=1}^n\pi(S_i,t(i,S^{i-1},M_p))\right]-\delta(\eps)
\vspace{-3pt}
\end{IEEEeqnarray*}
for sufficiently large $n$, where $\delta(\eps)\rightarrow0$ as $\eps\rightarrow0$.
\end{lemma}
\begin{proof}
Let $\eps>0$. Introduce the random variable \\$Q\sim\text{Unif}[1:n]$, independent of all other random variables present. First, we have
\begin{IEEEeqnarray*}{rCl}
 \IEEEeqnarraymulticol{3}{l}{
 I(S_Q;Z^n|M_pS^{Q-1}Q)}\\ \qquad
 &=& \frac1n \sum_{i=1}^nI(S_i;Z^n|M_pS^{i-1}) \\
 &=& \frac1n I(S^n;Z^n|M_p) \\
 &\leq& \frac1n I(M_sS^n;Z^n|M_p) \\
 &=& \frac1n I(M_s;Z^n|M_p)+\frac1n I(S^n;Z^n|M_pM_s) \\
 &=& \frac1n I(M_s;Z^n|M_p)\\
 &\leq& \frac1n I(M_s;Z^n,M_p)\\
 &<& \eps
\end{IEEEeqnarray*}
for sufficiently large $n$, where the final inequality follows the definition of $C_S$. Next, denote $P=P_{S^QZ^nM_pQ}$ and define the distribution 
\begin{equation*}
 \ovr{P}=P_{M_pS^{Q-1}Q}P_{S_Q|M_pS^{Q-1}Q}P_{Z^n|M_pS^{Q-1}Q}.
\end{equation*}
That is, $\ovr{P}$ is the markov chain $S_Q-M_pS^{Q-1}Q-Z^n$. Now, using Pinsker's inequality (see \cite{Cover}), we have
\begin{IEEEeqnarray*}{rCl}
\lVert P-\ovr{P} \rVert_{TV} &\leq& D(P\Vert \ovr{P})^{\frac12}\\
 &=& I(S_Q;Z^n|M_pS^{Q-1},Q)^{\frac12} \\
 \yesnumber \label{connect1}&<& \sqrt{\eps},
\end{IEEEeqnarray*}
where $\lVert P-Q \rVert_{TV}:=\sup_A|P(A)-Q(A)|$ is the variational distance between distributions $P$ and $Q$. Finally,
\begin{IEEEeqnarray*}{rCl}
\IEEEeqnarraymulticol{3}{l}{
 \min_{t(i,s^{i-1},z^n)}\Ebb\left[\frac1n \sum_{i=1}^n\pi(S_i,t(i,S^{i-1},Z^n))\right]}\\
 \yesnumber \label{connect_discussion} &\geq& \min_{t(i,s^{i-1},z^n,m_p)}\Ebb\left[\frac1n \sum_{i=1}^n\pi(S_i,t(i,S^{i-1},Z^n,M_p))\right] \\
 &=& \min_{t(i,s^{i-1},z^n,m_p)}\Ebb\left[\pi(S_Q,t(Q,S^{Q-1},Z^n,M_p))\right] \\
 \yesnumber \label{connect2}&\geq& \min_{t(i,s^{i-1},z^n,m_p)}\Ebb_{\ovr{P}}\left[\pi(S_Q,t(Q,S^{Q-1},Z^n,M_p))\right]-\delta(\eps)\IEEEeqnarraynumspace\\
 \yesnumber \label{connect3}&=& \min_{t(i,s^{i-1},m_p)}\Ebb_{\ovr{P}}\left[\pi(S_Q,t(Q,S^{Q-1},M_p))\right]-\delta(\eps)\\
 \yesnumber \label{connect4}&\geq& \min_{t(i,s^{i-1},m_p)}\Ebb\left[\pi(S_Q,t(Q,S^{Q-1},M_p))\right]-2\delta(\eps)\\
 &=& \min_{t(i,s^{i-1},m_p)}\Ebb\left[\frac1n\sum_{i=1}^n\pi(S_i,t(i,S^{i-1},M_p))\right]-2\delta(\eps),
\end{IEEEeqnarray*}
for sufficiently large $n$. The justification is as follows:
\begin{itemize}
 \item (\ref{connect_discussion}): Give Eve $M_p$ for free.
 \item (\ref{connect2}): Change the underlying distribution from $P$ to $\ovr{P}$ by using (\ref{connect1}) and Lemma~\ref{tv_expect}.
 \item (\ref{connect3}): Use Lemma~\ref{suffstat} since $\ovr{P}$ is a markov chain.
 \item (\ref{connect4}): Change the underlying distribution back to $P$.
\end{itemize}
\end{proof}
\begin{lemma}[Used in Lemma~\ref{connect} proof]
\label{suffstat}
Let $X,Y$, and $Z$ be random variables that form a markov chain $X-Y-Z$ and let $g$ be a function on $\Acal \times \Zcal$. Define two sets of functions, $F=\{f:\Xcal\times\Ycal\rightarrow \Acal\}$ and $F'=\{f:\Ycal\rightarrow\Acal\}$. Then
\begin{equation*}
 \min_{f\in F}\Ebb[g(f(X,Y),Z)]=\min_{f\in F'}\Ebb[g(f(Y),Z)].
\end{equation*}
\end{lemma}
Proof omitted.
\begin{lemma}[Used in Lemma~\ref{connect} proof]
\label{tv_expect}
 Let $f(x)$ be a bounded function, and $P$ and $Q$ pmfs on $\Xcal$. Then 
$$\Ebb_{P}[f(X)]\rightarrow \Ebb_Q[f(X)]\quad\text{as}\quad P\xrightarrow{TV}Q$$
\end{lemma}
\begin{proof}[Proof of Lemma~\ref{tv_expect}]
 \begin{IEEEeqnarray*}{rCl}
\IEEEeqnarraymulticol{3}{l}{
 \Ebb_{P}[f(X)]-\Ebb_Q[f(X)]}\\ 
 \qquad &=& \sum_{x\in\Xcal} (P(x)-Q(x))f(x)\\
 &\leq& f_{\max}\sum_x|P(x)-Q(x)|\\
 &=& 2\,f_{\max}\lVert P-Q \rVert_{TV}.
\end{IEEEeqnarray*}
\end{proof}

\section{\label{example} Example}
In this section, we give an analytical expression for the region in Theorem~\ref{innerbnd} when the value function $\pi(s,t)$ is Hamming distance and the broadcast channel is binary symmetric. We first consider the separate regions that we stated for source coding and channel coding (i.e., Theorems~\ref{bcc} and \ref{sourcethm}), then we join the results. Finally, we use the example to illustrate the difference between Theorems~\ref{innerbnd} and \ref{improved}.

Denoting the region in Theorem~\ref{bcc} by $\Rcal$, we have the following.
\begin{thm}
\label{channelex}
For a binary symmetric broadcast channel with $P_{Y|X}=\emph{BSC}(p_1)$, $P_{Z|X}=\emph{BSC}(p_2)$, and $p_1\leq p_2$, we have
\begin{equation*}
 \Rcal \hspace{-1pt}=\hspace{-5pt} \bigcup_{0\leq \gamma \leq \frac12}\left\{
 \begin{array}{l}
 \hspace{-3pt}(R_p,R_s):\\
 \hspace{-3pt}R_p< 1-h(\gamma * p_1)\\
 \hspace{-3pt}R_s< h(\gamma * p_1)-h(\gamma * p_1)-h(p_1)+h(p_2)\hspace{-2pt}
 \end{array}\hspace{-4pt}
 \right\}
\end{equation*}
where $h(\cdot)$ is the binary entropy function.
\end{thm}
The source coding achievability region in Theorem~\ref{sourcethm} is given as a union of $(R_p,R_s,\Pi)$ subregions, one subregion for each choice of $P_{U|S}$. However, once $P_{U|S}$ is fixed, $(R_p,R_s,\Pi)$ triples on the boundary satisfy $R_p=H(S)-R_s$, so we restrict our attention to $(R_s,\Pi)$ pairs. According to Theorem~\ref{sourcethm}, the boundary is given by
\begin{equation}
\label{secrpayfunc}
 \Pi(R_s)=\max_{\substack{P_{U|S}:\\ H(S|U)=R_s}}\min_{t(u)}\Ebb[\pi(S,t(U))].
\end{equation}
When $\pi(s,t)$ is Hamming distance, we give an analytical expression for (\ref{secrpayfunc}).
\begin{thm}
 \label{sourceex}
 Fix $P_S$ and let $\pi(s,t)=1_{\{s\neq t\}}$. Define the function $f(R)$ as the linear interpolation of the points $(\log n,\frac{n-1}{n}),n=1,2,\ldots$ and the function $g(R)$ as the constant $1-\max_s P_S(s)$. Then (\ref{secrpayfunc}) is given by
\begin{equation*}
 \Pi_H(R_s)=\min(f(R_s),g(R_s)).
\end{equation*}
\end{thm}
We now use Theorems~\ref{channelex} and \ref{sourceex} to give an expression for Theorem~\ref{innerbnd} when the value function is Hamming distance and the broadcast channel is binary symmetric. 
\begin{thm}
 \label{inner_ex}
 With source $P_S$, channels $P_{Y|X}=\emph{BSC}(p_1)$ and $P_{Z|X}=\emph{BSC}(p_2)$, and value function $\pi(s,t)=1_{\{s\neq t\}}$, an achievable payoff is
 \begin{equation*}
\Pi=\begin{cases}
 \Pi_H(h(\gamma * p_1)-h(\gamma * p_2) - h(p_1)+h(p_2)) \\
  \qquad \text{if } h(p_2)-h(p_1) < H(S) < 1-h(p_1)\\
 \Pi_H(H(S)) \\
  \qquad \text{if } H(S)\leq h(p_2)-h(p_1)
\end{cases}
 \end{equation*}
where $\gamma\in[0,\frac12]$ solves $H(S)=1-h(\gamma * p_2) - h(p_1)+h(p_2)$ and $\Pi_H(R_s)$ is as in Theorem~\ref{sourceex}.
\end{thm}
\begin{figure}
 \centering
 \scalebox{0.58}{
 \input{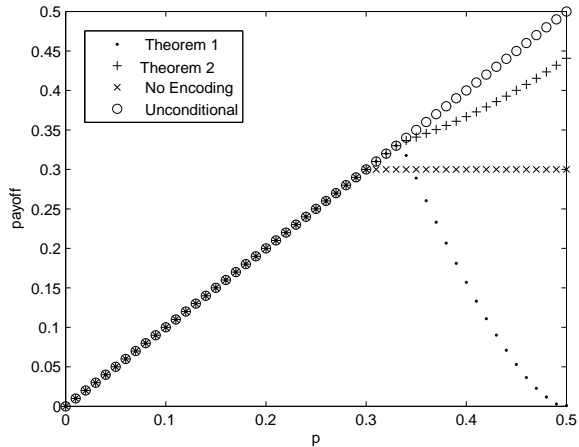}
 }
 \caption{The result of Thm \ref{inner_ex} with $P_S=\text{Bern}(p)$, $p_1=0$, and $p_2=0.3$. The ``Theorem~\ref{innerbnd}'' curve is directly Thm~\ref{inner_ex} and ``Theorem~\ref{improved}'' is the improvement to Thm~\ref{innerbnd}. ``Unconditional'' is the payoff when Eve only knows $P_S$, while ``No Encoding'' refers to $X^n=S^n$.}
\label{secrecy_plot}
\end{figure}
\subsection*{Discussion of Theorem~\ref{improved}}
When we established the lower bound in Theorem~\ref{innerbnd}, we did so with the structure of a public channel and secure channel in mind; however, as mentioned, the public channel may not be truly public. We made the assumption that the adversary is freely given $M_p$; indeed, the proof of Lemma~\ref{connect} illustrates this in (\ref{connect_discussion}), where we suffer a loss in our payoff analysis by including $M_p$ as an input to the adversary's strategy. 

We can strengthen the analysis of Theorem~\ref{innerbnd} by taking into account the equivocation of the public message. For blocklength $n$, the equivocation of the public message vanishes at a certain time $k$ due to the adversary's ongoing accumulation of past source symbols $\hat{S}^{k-1}$. Before time $k$, the payoff is $\Pi_{\max} = \min_t\Ebb[\pi(S,t)]$ (i.e., the unconditional payoff). After time $k$, the payoff is as in Theorem~\ref{innerbnd}. Denoting this payoff by $\Pi_1$, we can now achieve (Theorem~\ref{improved})
\begin{equation*}
\Pi_2 = \frac{k}{n} \Pi_{\max} + \left(1-\frac{k}{n}\right)\Pi_1.
\end{equation*}
The ratio $\frac{k}{n}$ is found to be
\begin{equation*}
 \frac{k}{n} \approx \frac{[I(V;Y)-I(V;Z)]^+}{I(S;U)}.
\end{equation*}
Figure \ref{secrecy_plot} shows the difference between Theorem~\ref{innerbnd} and Theorem~\ref{improved}, as well as a comparison to the curves that correspond to no encoding and unconditional payoff. Unconditional payoff refers to the distortion that Eve suffers if her only knowledge is the source distribution, and no encoding refers to simply taking the source as the input to the channel and bypassing the encoder. The example is for a Bernoulli source with bias $p$. If we assume that Eve has full knowledge of the public message, then we see that for, say, $p=0.4$, the distortion guaranteed by the weaker theorem is even worse than if no encoding was used. This illustrates the importance of Eve's equivocation of the public message. 
\section{Upper Bound}

The upper bound is established by using ideas from the converses in \cite{Cuff2010} and \cite{Csiszar--Korner1978}.
\begin{thm}
\label{outerbnd}
Fix $P_S$, $P_{YZ|X}$, and $\pi(s,t)$. If a payoff $\Pi$ is achievable, then the inequalities
\begin{IEEEeqnarray*}{rCl}
 H(S) &\leq& I(W;Y)\\
 H(S|U) &\leq& [I(W;Y|V)-I(W;Z|V)]^+\\
 \Pi &\leq& \min_{t(u)}\Ebb[\pi(S,t(U))]
\end{IEEEeqnarray*}
must hold for some distribution\\ $P_SP_{U|S}P_VP_{W|V}P_{X|W}P_{YZ|X}$.
\end{thm}
\smallskip
\section{Conclusion}
\smallskip
By considering source and channel operations separately, we have given results on how well communication systems can perform against a worst-case adversary when the secrecy resource is physical-layer security and the adversary has causal access to the source. We have seen that a guarantee of weak secrecy can be used in conjunction with our operationally-relevant measure of secrecy.  

\section{Acknowledgements}
This research was supported in part by the National Science Foundation under Grants CCF-1016671, CCF-1116013, and CCF-1017431, and also by the Air Force Office of Scientific Research under Grant FA9550-12-1-0196.

\bibliographystyle{unsrt}

\end{document}